\documentclass[a4,11pt]{article}

\usepackage{algorithm}
\usepackage{algorithmicx}
\usepackage{algpseudocode}
\algnewcommand\Input{\item[\textbf{Input:}]}%
\algnewcommand\Output{\item[\textbf{Output:}]}%
\usepackage{amsmath}% sophisticated mathematical formulas with amstex (includes \text{})
\usepackage{mathtools, commath}% fix amsmath deficiencies
\usepackage{amssymb}% sophisticated mathematical symbols with amstex
\usepackage{amsthm}% theorem environments
\usepackage[american]{babel}%
\usepackage{bm}% for bold math symbols
\usepackage{bbm}% only for indicator functions (replace by dsfont and \mathds{1} if type 3 fonts are not allowed)
\usepackage{booktabs}
\usepackage[top=30truemm,bottom=30truemm,left=25truemm,right=25truemm]{geometry}
\usepackage[bookmarksnumbered=true]{hyperref}
\hypersetup{
	colorlinks = true,
	linkcolor = blue,
	anchorcolor = blue,
	citecolor = blue,
	filecolor = blue,
	urlcolor = blue
}
\usepackage{comment}% for commenting out
\usepackage{color} % for textcolor command
\usepackage{csquotes}% recommended by biblatex
\usepackage{eso-pic}%
\usepackage{etoolbox}%
\usepackage{fancyvrb}% for listings
\usepackage[T1]{fontenc}%
\usepackage{graphicx}
\usepackage{caption}
\usepackage{subcaption}
\usepackage{float}
\usepackage{grffile}%
\usepackage[
hyperref,%
table%
]{xcolor}%
\usepackage{ifthen}
\newboolean{PrintVersion}
\setboolean{PrintVersion}{false}
\ifthenelse{\boolean{PrintVersion}}{   % for improved print quality, change some hyperref options
	\hypersetup{  % override some previously defined hyperref options
		citecolor=black,%
		filecolor=black,%
		linkcolor=black,%
		urlcolor=black}
}{} % end of ifthenelse (no else)

\usepackage{latexsym,indentfirst}
\usepackage{listings}% for including source code
\usepackage{lmodern}%
\usepackage{lscape} % for rotation
\usepackage{multirow}%
\usepackage{microtype}%
\usepackage[round]{natbib}
\usepackage{rotating}%
\usepackage{scrlayer-scrpage}%
\usepackage{siunitx}%
\usepackage{tabularx}%
\usepackage{threeparttable} %for table partitioned
\usepackage{tikz}% sophisticated graphics package
\usetikzlibrary{positioning}% for relative positioning of nodes
\usetikzlibrary{calc}
\usetikzlibrary{shapes.geometric}
\usetikzlibrary{arrows.meta}
\usepackage{vruler}% for showing line numbers
\usepackage{wrapfig}%
\usepackage{enumerate}
\usepackage{cancel}

\theoremstyle{definition}
\newtheorem{theorem}{Theorem}[section]

\newtheorem{assumption}[theorem]{Assumption}

\xdefinecolor{lightgray}{RGB}{247, 247, 247}%
\xdefinecolor{semilightgray}{RGB}{240, 240, 240}%
\xdefinecolor{middlegray}{RGB}{127, 127, 127}%
\xdefinecolor{deepskyblue}{RGB}{0, 154, 205}%
\xdefinecolor{chocolate}{RGB}{205, 102, 29}%

\newcommand{\EE}{{\mathbb{E}}}

\newcommand{\II}{{\mathbb{I}}}

\newcommand{\RR}{{\mathbb{R}}}

\newcommand{\n}{{\noindent}}

\definecolor{uared}{rgb}{0.85, 0.0, 0.3}

\title{Semiparametric copula-based quantile regression for semicontinuous outcomes with application to healthcare data}

\author{Guanjie Lyu\footnote{ %Corresponding author at: 
		Department of Mathematics and Statistics, University of Windsor, Canada, E-mail: lvg@\,uwindsor.ca} ,\
	Mohamed Belalia\footnote{Department of Mathematics and Statistics, University of Windsor, Canada, E-mail: Mohamed.Belalia@uwindsor.ca}
	 \ and\
	Abdulkadir Hussein\footnote{Department of Mathematics and Statistics, University of Windsor, Canada, E-mail: ahussein@uwindsor.ca}
}

\begin{document}
	%\gtfamily
	\date{\vspace{-5ex}}
	
	\maketitle

	%%%%%%    TEXT START    %%%%%%
	\begin{abstract}
A semiparametric copula-based two-part quantile regression framework is developed for the analysis of semicontinuous outcomes characterized by a point mass at zero and a continuous positive component. The proposed approach models the occurrence and magnitude processes separately and links them through copula-based conditional distributions, allowing for flexible dependence structures and nonlinear covariate effects across quantiles. Large-sample properties of the resulting estimator are established, and extensive simulation studies demonstrate improved finite-sample performance relative to logistic/linear quantile regression, particularly under nonlinear dependence and substantial zero inflation. An application to healthcare data illustrates how the proposed method provides a nuanced characterization of the association between social deprivation and uncompensated and charity care burdens, revealing heterogeneous and nonlinear effects that are not captured by competing approaches.
	\end{abstract}
	
	% JEL classification and keywords
	\hspace{0mm}\\
%	\emph{MSC classification:} 62H20, 62G20, 62H10, 62H12\\

%---------------------------------------------------------------------	
	\section{Introduction}\label{sec:introduction}
%---------------------------------------------------------------------	

% description difference between count data and semicontinuous data

% description of existent models, consult with Dr. Hussein

% copula based model comparison, the advantages on non-crossing and monotonicity

% discrete vine or low dimension with parameter estimation in Nasri2023b

\n Quantile regression, introduced by~\citet{Koenker1978}, provides a flexible framework for modeling the conditional distribution of a response variable. Unlike classical mean regression, quantile regression estimates conditional quantiles, enabling researchers to explore the effects of covariates across the entire distribution of the outcome. This is particularly useful when the response variable exhibits heteroscedasticity or when we are interested in the tails of the distribution, such as extreme values. Quantile regression has gained traction in fields such as economics, finance, and medicine, where the conditional distribution's shape and spread can reveal important insights beyond the mean.

In recent years, the integration of copulas into quantile regression has attracted significant attention, especially for modeling complex dependency structures between variables. Copulas are powerful tools that allow for the separate modeling of marginal distributions and the dependence between variables. This makes copula-based methods particularly useful in cases where the assumption of independence between covariates and errors is violated, or where multivariate dependence needs to be captured flexibly. Building on the semiparametric copula-based quantile regression framework of \cite{Noh2015semiparametric}, this line of work was further extended to accommodate censored outcomes by \cite{De2017CopulaQR}. Independently, \cite{Kraus2017} pioneered the use of D-vine copula constructions for quantile regression, demonstrating how flexible pair-copula decompositions can substantially improve the estimation of conditional quantiles. Related contributions by \cite{Remillard2017} developed copula-based conditional quantile estimators that explicitly model the dependence structure between covariates and responses, yielding efficiency gains over purely marginal approaches. 

A challenging problem in regression modeling arises when the data exhibit clumping at zero, which is often referred to as zero-inflated or semicontinuous data. These types of data are prevalent in fields such as health economics, where many individuals might have zero healthcare expenditures, while the rest have continuous positive expenditures. In environmental studies, zero-inflation is common when measuring pollutant levels or rainfall, where many observations are zeros (no pollution or no rain), followed by a continuous range of positive values. Traditional regression models struggle to handle this type of data due to its dual structure: a binary process governing the occurrence of zeros, and a continuous process for the positive values.~\cite{Neelon2016} provided a comprehensive review of models for zero-inflated data, focusing on health services research, and introduced zero-modified count models that combine logistic regression for modeling the probability of zeros with a continuous distribution for the positive values. Similarly,~\citet{Heras2018} and~\citet{Ling2022} extended zero-inflated modeling to quantile regression, allowing for the estimation of linear quantile framework in the presence of zero-inflation.

Despite these advancements, existing zero-inflated quantile regression models have limitations, particularly when handling complex dependency structures in the data. The copula-based framework offers a more flexible approach by separately modeling the binary part (i.e., the probability of a non-zero response) and the continuous part (i.e., the conditional quantiles of the positive outcomes), and then linking these components through copulas. Copulas allow for more accurate modeling of the dependence structure between covariates and both the binary and continuous components, thus improving the overall fit and predictive performance of the model.

In this manuscript, we propose a semiparametric copula-based quantile regression model specifically designed for semicontinuous data. Our approach builds on the work of \cite{Mesfioui2023}, who developed copula-based models for binary regression, and \cite{Kraus2017}, who applied copulas to quantile regression. By combining these methods, we model the probability of a non-zero outcome using a copula-based binary regression model and the positive conditional quantiles using a copula-based quantile regression model. This two-part approach allows us to capture the unique structure of semicontinuous data while flexibly accounting for the dependence between covariates and the outcome.

The motivation for this method comes from real-world datasets commonly found in healthcare, insurance, and environmental studies. For instance, in healthcare expenditure data, many patients incur zero costs during a period, while a smaller proportion incurs varying positive expenditures. In insurance data, many policyholders file no claims, but those who do file claims exhibit a wide range of claim sizes. The proposed model addresses these challenges by accurately capturing both the probability of a non-zero outcome and the distribution of positive outcomes, making it a valuable tool for analyzing data with clumping at zero.

The remainder of the paper is organized as follows. Section~\ref{sec:method} describes the proposed methodology. Section~\ref{sec:simulation} presents results from a simulation study, while 
Section~\ref{sec:application} illustrates the approach using a healthcare data application. Finally, Section~\ref{sec:conclusion} provides concluding remarks.

%---------------------------------------------------------------------	
\section{The proposed method}\label{sec:method}
%---------------------------------------------------------------------	

\n This section develops a semiparametric copula-based quantile regression framework for semicontinuous outcomes. 
The proposed approach explicitly accounts for the mixture structure induced by a point mass at zero and a continuous positive component, resulting in a conditional quantile function that exhibits an inherent discontinuity. 
By modeling the occurrence and magnitude of positive responses separately and linking them through copula-based conditional distributions, the framework accommodates flexible dependence structures while preserving tractable large-sample properties.

\subsection{Semicontinuous distributions and conditional quantiles}\label{sec:semicts distribution}

\n Let $Y$ be a semicontinuous outcome supported on $[0,\infty)$ with a point mass at zero and a
continuous distribution on $(0,\infty)$. Let $\bm X=(X_1,\ldots,X_d)^\top$ denote a vector of
continuous covariates. For a fixed covariate value $\bm x$, define the positive component
$(Y\mid Y>0)$ with conditional distribution function
\[
F_{Y\mid Y>0}(y\mid \bm x)
=
P(Y\le y \mid \bm X=\bm x,\,Y>0), \qquad y>0,
\]
which is assumed to be absolutely continuous with density
$f_{Y\mid Y>0}(\cdot\mid \bm x)$ on $(0,\infty)$. The conditional distribution function of $Y$ given $\bm X=\bm x$ admits the mixture representation, for $y\ge 0$, 
\begin{align}\label{eq:mix-cdf}
F_Y(y\mid \bm x)
&=
P(Y\le y\mid \bm X=\bm x)\notag\\
&=
P(Y=0\mid \bm X=\bm x)
+ P(Y>0\mid \bm X=\bm x)\,F_{Y\mid Y>0}(y\mid \bm x),\notag\\
&= \pi_{0}(\bm{x}) + \{1-\pi_0(\bm x)\}\,F_{Y\mid Y>0}(y\mid \bm x).
\end{align}
For $\tau\in(0,1)$, define the conditional $\tau$-quantile of $Y$ given $\bm X=\bm x$ by the
generalized inverse~\citep{Embrechts2013Inverse}, i.e.,
\[
Q_Y(\tau\mid \bm x)
=
\inf\{y\ge 0:\ F_Y(y\mid \bm x)\ge \tau\}.
\]
Note that, due to the point mass at zero, the conditional quantile function satisfies
\begin{equation}\label{eq:zero_quantile}
Q_Y(\tau\mid \bm x)=0,
\qquad \forall \; 0<\tau\le \pi_{0}(\bm{x}).
\end{equation}
%For $\tau>P(Y=0\mid \bm X=\bm x)$, the quantile is determined by the positive component.
Further, define the
scaled quantile level
\begin{equation}\label{eq:tau-scaled}
\tau_s(\tau,\bm x)
=
\frac{\tau-\pi_{0}(\bm{x})}{1-\pi_{0}(\bm{x})}.
\end{equation}
From Equation~\eqref{eq:mix-cdf}, it follows that, for $\tau>\pi_{0}(\bm{x})$,
\[
F_{Y\mid Y>0}\!\left(Q_Y(\tau\mid \bm x)\mid \bm x\right)
=
\tau_s(\tau,\bm x),
\]
and hence the conditional quantile function is given by
\begin{equation}\label{eq:semicont-quantile}
Q_Y(\tau\mid \bm x)
=
Q_{Y\mid Y>0}\!\left(\tau_s(\tau,\bm x)\mid \bm x\right),
\end{equation}
where
\[
Q_{Y\mid Y>0}(\alpha\mid \bm x)
=
\inf\{y>0:\ F_{Y\mid Y>0}(y\mid \bm x)\ge \alpha\}
\]
denotes the conditional quantile function of the positive component. 

Combining Equations~\eqref{eq:zero_quantile}--\eqref{eq:semicont-quantile} shows that estimation of
$Q_Y(\tau\mid \bm x)$ reduces to modeling
$\pi_{0}(\bm{x})$ and the positive-part conditional quantile function
$Q_{Y\mid Y>0}(\cdot\mid \bm x)$.
In the next section, these two components are linked to covariates through copula-based conditional
distributions, yielding a flexible semiparametric framework for semicontinuous quantile regression.

\subsection{Copula based quantile regression model}

\n Copulas provide a flexible framework for modeling dependence structures independently of marginal
distributions. Let $(\Lambda,\bm W)$ be a $(d+1)$-dimensional random vector with continuous marginal
distribution functions $F_\Lambda$ and $F_{\bm W}=(F_{W_1},\ldots,F_{W_d})^\top$, and define the probability
integral transforms
\[
\Lambda^\ast = F_\Lambda(\Lambda), \qquad \bm W^\ast = F_{\bm W}(\bm W).
\]
Then $(\Lambda^\ast,\bm W^\ast)$ has support on $[0,1]\times[0,1]^d$ and admits a copula representation.
By Sklar’s theorem~\citep{Sklar1959}, there exists a copula $\mathcal C:[0,1]\times[0,1]^d\to[0,1]$
such that
\begin{equation}\label{eq:sklar-general}
P(\Lambda^\ast \le \lambda,\ \bm W^\ast \le \bm w) = \mathcal C(\lambda,\bm w),
\qquad \forall\;(\lambda,\bm w)\in[0,1]\times[0,1]^d.
\end{equation}
When $\mathcal C$ is absolutely continuous, the conditional distribution of $\Lambda^\ast$ given
$\bm W^\ast=\bm w$ is well defined and given by
\[
\mathcal C_{\Lambda^\ast\mid \bm W^\ast}(\lambda\mid\bm w)
=
\frac{\partial^d \mathcal C(\lambda,\bm w)/\partial w_1\cdots\partial w_d}
{\partial^d \mathcal C(1,\bm w)/\partial w_1\cdots\partial w_d},
\qquad \forall\;\lambda\in[0,1].
\]
In the semicontinuous setting considered here, the conditional quantile function derived in
Section~\ref{sec:semicts distribution} naturally decomposes into two components:
the occurrence probability $1-\pi_0(\bm x)$ and the positive-part conditional quantile
$Q_{Y\mid Y>0}(\cdot\mid \bm x)$. By convention~\citep[see for example][]{Neelon2016a}, we consider that the semicontinuous data emerges from two distinct stochastic process: one governing the occurrence of zero values,  and the second determining the observed value when a non-zero response is present. Therefore,  we model these two components separately using copula-based conditional distributions, allowing
for potentially different dependence structures in the binary and continuous parts of the model.
Specifically, the occurrence component is modeled using the copula-based
binary regression framework of~\citet{Mesfioui2023}. Let
\[
Z = \mathbbm 1(Y>0), \qquad p_0 = P(Z=0)=P(Y=0),
\]
where $\II(\cdot)$ denotes the indicator function and define the probability integral transforms
\[
V = F_Z(Z), \qquad \bm U = F_{\bm X}(\bm X).
\]
Then $\bm U\in[0,1]^d$ summarizes the covariates on the uniform scale, while the binary response is
carried by $Z$ (equivalently by $V$). The conditional probability of a positive response is given by
\begin{equation}\label{eq:binary-copula}
1-\pi_0(\bm x)
=
\Pr(Z=1 \mid \bm X=\bm x)
=
1-
C_{V\mid \bm U}\!\left(
p_0 \,\middle|\, F_{\bm X}(\bm x);\boldsymbol\theta_1
\right),
\end{equation}
where \( C_{V\mid \bm U}(\cdot \mid \bm u; \boldsymbol\theta_1) \) denotes the conditional distribution function of \(V\) given \( \bm U=\bm u \), induced by a \((d+1)\)-dimensional copula \(C\) that links the binary outcome and the \(d\)-dimensional covariate vector. The parameter vector \( \boldsymbol\theta_1 \in \Theta_1 \subset \mathbb{R}^q \) governs the dependence structure. Further, the positive-part conditional quantile
$Q_{Y\mid Y>0}\!\left(\tau_s(\tau,\bm x)\mid \bm x\right)$
is modeled using copula-based quantile regression with copula function $D$.
To distinguish this component from the binary regression part, we introduce the probability
integral transforms
\[
V^+ = F_{Y\mid Y>0}(Y\mid \bm X), \qquad
\bm U^+ = F_{\bm X\mid Y>0}(\bm X),
\]
where $V^+\in[0,1]$ corresponds to the positive response and
$\bm U^+\in[0,1]^d$ represents the concomitant covariates associated with $Y\mid Y>0$.
The conditional quantile function is then given by
\begin{equation}\label{eq:positive-copula-qr}
Q_{Y\mid Y>0}\!\left(\tau_s(\tau,\bm x)\mid \bm x\right)
=
F_{Y\mid Y>0}^{-1}\!\left(
D_{V^+\mid \bm U^+}^{-1}\!\left(
\tau_s(\tau,\bm x)\mid F_{\bm X\mid Y>0}(\bm x); \bm\theta_2
\right)
\right),
\end{equation}
where $D_{V^+\mid \bm U^+}(\cdot\mid \bm v;\bm\theta_2)$ denotes the conditional distribution
function induced by a $(d+1)$-dimensional copula $D$ linking the positive response and its
concomitant covariates, and $\bm\theta_2\in\Theta_2\subset\mathbb R^p$ is the associated copula
parameter vector.
See~\citet{Kraus2017} and~\citet{Remillard2017} for further details on copula-based quantile
regression.

With all ingredients at hand, combining
Equations~\eqref{eq:binary-copula} and \eqref{eq:positive-copula-qr},
the two-part conditional quantile function of $Y$ given $\bm X=\bm x$
is given, for any $\tau\in(0,1)$, by
\begin{align}\label{eq:2023-12-13, 8:56PM}
	 &Q_Y(\tau\mid \bm x)
=
\II\!\left\{\tau>\pi_0(\bm x)\right\}
\,
Q_{Y\mid Y>0}\!\left(\tau_s(\tau,\bm x)\mid \bm x\right),\\
	 \text{where} \quad &
	 \begin{cases}
\pi_0(\bm x)
=
C_{V\mid \bm U}\!\left(p_0\mid F_{\bm X}(\bm x);\bm\theta_1\right),\\
Q_{Y\mid Y>0}\!\left(\tau_s(\tau,\bm x)\mid \bm x\right)
=
F_{Y\mid Y>0}^{-1}\!\left(
D_{V^+\mid \bm U^+}^{-1}\!\left(
\tau_s(\tau,\bm x)
\mid F_{\bm X\mid Y>0}(\bm x);\bm\theta_2
\right)
\right).\notag
	 \end{cases}
     \end{align}
\subsection{{The proposed estimator}}

\n To construct the two-part conditional quantile function in
Equation~\eqref{eq:2023-12-13, 8:56PM}, all unknown quantities appearing in the
representation must be replaced by suitable estimators.
The proposed estimation procedure consists of the following considerations:
\begin{enumerate}[(a)]
\item
The jump location in the conditional quantile function is determined by
$p_0 = P(Y=0)$ and the marginal distribution of the covariates $F_{\bm X}$.
These quantities are estimated by their empirical distribution functions,
yielding the estimators $\widehat p_0$ and $\widehat F_{\bm X}$, respectively.

\item
For the positive component, the conditional distribution function
$F_{Y\mid Y>0}$ and the concomitant covariate distribution
$F_{\bm X\mid Y>0}$ are estimated using kernel smoothing methods, resulting in
the estimators $\widehat F_{Y\mid Y>0}$ and $\widehat F_{\bm X\mid Y>0}$.
Kernel smoothing is employed to ensure numerical stability of the inverse
distribution function appearing in the conditional quantile representation.
Bandwidth parameters may be selected using standard data-driven procedures,
such as plug-in or cross-validation methods; see, for example,~\cite{Azzalini1981note},~\cite{ Sarda1991estimating}, and~\cite{Bowman1998bandwidth}.

\item
The dependence parameter vectors $\boldsymbol\theta_1$ and
$\boldsymbol\theta_2$, governing the binary-part copula $C$ and the
positive-part copula $D$, respectively, are estimated using
pseudo-likelihood methods based on the corresponding copula densities.
This approach is adopted for its computational efficiency and favorable
large-sample properties in semiparametric copula models; see, for example,~\cite{Genest1995},~\cite{Joe2005asymptotic}, and~\cite{Chen2006estimation}.
Alternative estimation strategies, such as inference functions for margins~\citep{Joe1996IFM} and simulated method of moments~\citep{Oh2013simulated, Belalia2024generalized}, are also available in this context
but are not pursued here.
\end{enumerate}
Moreover, following the approach of \cite{Ling2022}, we introduce a
piecewise estimator for $Q_Y(\tau\mid\bm x)$ to accommodate the jump
discontinuity as
$\tau$ goes closer to $
C_{V\mid\bm U}\!\left(\widehat p_0 \mid \widehat F_{\bm X}(\bm x);
\widehat{\boldsymbol\theta}_1\right)$. Let $\delta \in (0, 1/2)$, partitioning the support of quantile level $(0, 1)$ into three sub-intervals $A_{1n}, A_{2n}$ and $A_{3n}$:
 \begin{align*}
 	A_{1n} &= \left\{\tau: 0<\tau< {C}_{V|\bm{U}}\left(\widehat{p}_0|\widehat{F}_{\bm{X}}(\bm{x}); \widehat{\bm{\theta}}_1\right)\right\}, \\
 	A_{2n} &= \left\{\tau: {C}_{V|\bm{U}}\left(\widehat{p}_0|\widehat{F}_{\bm{X}}(\bm{x}); \widehat{\bm{\theta}}_1\right) \leq \tau \leq {C}_{V|\bm{U}}\left(\widehat{p}_0|\widehat{F}_{\bm{X}}(\bm{x}); \widehat{\bm{\theta}}_1\right)+n^{-\delta}\right\}, \\
 A_{3n}&= \left\{{C}_{V|\bm{U}}\left(\widehat{p}_0|\widehat{F}_{\bm{X}}(\bm{x}); \widehat{\bm{\theta}}_1\right)+n^{-\delta}<\tau <1\right\}.
 \end{align*}
Then the estimator $\widehat{Q}_Y(\tau |\bm{x})$ is defined as zero on $A_{1n}$ and uses linear interpolation if  $\tau \in A_{2n}$ resulting a piecewise function defined by
\begin{align}\label{eq:2024-06-07, 9:49PM}
\widehat{Q}_Y(\tau |\bm{x})&=0\cdot \II(\tau \in A_{1n}) + \widehat{F}_{Y|Y>0}^{-1}\left({D}_{V^+|\bm{U}^+}^{-1}\left(\widehat{\tau}_s\left({C}_{V|\bm{U}}\left(\widehat{p}_0|\widehat{F}_{\bm{X}}(\bm{x}); \widehat{\bm{\theta}}_1\right)+n^{-\delta}, \bm{x}\right)\big|\widehat{F}_{\bm{X}|Y>0}(\bm{x}); \widehat{\bm{\theta}}_2\right)\right)\notag \\
&\quad  \cdot \frac{\tau-{C}_{V|\bm{U}}\left(\widehat{p}_0|\widehat{F}_{\bm{X}}(\bm{x}); \widehat{\bm{\theta}}_1\right)}{n^{-\delta}}\cdot \II(\tau \in A_{2n})\notag \\
&\quad + \widehat{F}_{Y|Y>0}^{-1}\left({D}_{V^+|\bm{U}^+}^{-1}\left(\widehat{\tau}_s\left(\tau, \bm{x}\right)\big|\widehat{F}_{\bm{X}|Y>0}(\bm{x}); \widehat{\bm{\theta}}_2\right)\right) \cdot \II(\tau \in A_{3n}).
\end{align}
The piecewise construction, together with a shrinking interpolation region
of width $n^{-\delta}$, ensures numerical stability and avoids erratic
behavior of the inverse conditional distribution in a neighborhood of the
jump. Under the regularity conditions stated in Section~2.4, the estimators
introduced above are consistent for their population counterparts, which
justifies the plug-in construction of $\widehat Q_Y(\tau|\bm x)$. For ease of notation, we denote this estimator by $\widehat Q(\tau|\bm x)$
in the remainder of the paper.

\subsection{Weak consistency}

\n This subsection establishes the weak consistency of the proposed two-part copula-based conditional quantile estimator. The main technical challenge arises from the intrinsic discontinuity of the conditional quantile function induced by the point mass at zero, which leads to a nonstandard jump in the quantile process. To address this issue, we adopt a piecewise construction of the estimator around the jump location and study its asymptotic behavior separately on each region of the quantile index. The assumptions stated below ensure regularity of the positive component, smoothness of the copula-based conditional distributions, and root-$n$ consistency of the copula parameter estimators. 

\begin{assumption}\label{ass:2024-06-10, 10:11AM}
The following assumptions are required for the consistency of the proposed estimator $\widehat{Q}(\tau|\bm{x})$.
\begin{enumerate}[(i)]

		\item For any  $\bm{X}=\bm{x}$, 
		\begin{equation}\label{eq:2023-12-13, 8:46PM}
			\lim_{\tau \downarrow 0} Q_Y(\tau | \bm{X}=\bm{x}, Y > 0) = 0.
		\end{equation}

		 \item The conditional distribution function $F_Y(\cdot| \bm{x}, Y>0)$ is absolutely continuous with a positive continuous density $f_{Y|Y>0}(\cdot|\bm{x})$ on $[0,  \infty)$.

		\item $C_{V|\bm{U}} (\cdot ;\bm{\theta}_1)$ is twice continuously differentiable in a neighbourhood of the true parameter $\bm{\theta}_1$ and it is bounded away from zero.
		
		\item The Fisher information matrix $I(\bm{\theta}_1)=\EE\left[-\frac{\partial^2 \log \ell(\bm{\theta}_1, p_0, F_{\bm{X}}(\bm{X}))}{{\partial \bm{\theta}_1\partial \bm{\theta}_1}^\top}\right]$ exists and it is non-singular.
		
		\item $\frac{\partial^2 \log C_{V|\bm{U}}(v|\bm{u}; \bm{\theta}_1)}{\partial v\partial \theta_{1k}}$ and $\frac{\partial^2 \log C_{V|\bm{U}}(v|\bm{u}; \bm{\theta}_1)}{\partial u_j\partial \theta_{1k}}$ are continuous on $\Theta_1 \times (0, 1)\times (0, 1)^d$ for $k=1, \ldots, q, j=1, \ldots, d$ and $\Theta_1\subset \RR^q$.

		\item  $D_{V|\bm{U}} (\cdot ;\bm{\theta}_2)$ is twice continuously differentiable in a neighbourhood of the true parameter $\bm{\theta}_2$ and it is bounded away from zero.
		
		\item The gradients $\nabla_{\bm{\theta}_2}D_{V|\bm{U}}(v|\bm
		u; \bm{\theta}_2)$ and $\nabla_{\bm{u}} D_{V|\bm{U}}(v|\bm
		u; \bm{\theta}_2)$ are continuous.
		
		\item $\widehat{\bm{\theta}}_2-\bm{\theta}_2=n^{-1}\sum_{i=1}^n\bm{\eta}_i+o_p(n^{-1/2})$, where $\bm{\eta}_i=\bm{\eta}(V_i, \bm{U}_i; \bm{\theta}_2)$ is a $p$-dimensional  random vector with zero mean and finite variance.
\end{enumerate}
\end{assumption}

\hyperref[ass:2024-06-10, 10:11AM]{Assumption~\ref{ass:2024-06-10, 10:11AM}} (i) ensures the continuity of the conditional quantile function $Q_Y(\tau|\bm {X}=\bm{x}, Y>0)$ at zero. \hyperref[ass:2024-06-10, 10:11AM]{Assumption~\ref{ass:2024-06-10, 10:11AM}} (ii) is necessary to guarantee that the conditional quantile estimator has a finite variance. \hyperref[ass:2024-06-10, 10:11AM]{Assumption~\ref{ass:2024-06-10, 10:11AM}} (iii)--(v) are assumptions for the copula-based binary regression model.  \hyperref[ass:2024-06-10, 10:11AM]{Assumption~\ref{ass:2024-06-10, 10:11AM}} (vi)--(viii) are assumptions for the copula-based quantile regression model.

Under these assumptions, the following theorem establishes the weak consistency of the proposed two-part conditional quantile estimator $\widehat{Q}(\tau\mid\bm{x})$. 

\begin{theorem}\label{thm:2024-06-09, 8:46PM}
Under \hyperref[ass:2024-06-10, 10:11AM]{Assumption~\ref{ass:2024-06-10, 10:11AM}}, we have, for $\tau \in (0, 1)$ and $\bm{x}\in \RR^d$,
\begin{equation*}
\widehat{Q}\left(\tau \mid \bm{x}\right)\xrightarrow[n\to \infty]{P} {Q}\left(\tau \mid \bm{x}\right),
\end{equation*}
where ``$\xrightarrow[n\to \infty]{P}$'' stands for convergence in probability.
\end{theorem}

\begin{proof}
The proof is provided in \hyperref[app]{Appendix}.
\end{proof}

%---------------------------------------------------------------------	
\section{Simulation study}\label{sec:simulation}
%---------------------------------------------------------------------	

\n This section investigates the finite-sample performance of the proposed copula-based two-part quantile regression estimator through a comprehensive simulation study. 
The simulations are designed to assess estimation accuracy across different quantile levels, degrees of zero inflation, and dependence structures, and to compare the proposed method with existing logistic/linear quantile regression approaches. 
Performance is evaluated using integrated mean squared error (IMSE) and its bias--variance decomposition, allowing for a detailed examination of the trade-off between model flexibility and stability.

\subsection{Simulation setting}

\n To evaluate the finite-sample performance of the proposed copula-based two-part quantile regression estimator, we conduct a simulation study under a variety of data generating processes (DGPs). 
The DGPs are designed to examine the effect of different dependence structures, association strengths, and zero-inflation mechanisms, and to facilitate comparison with a competing logistic/linear quantile regression model.

For the copula-based DGPs, we consider a single continuous covariate and generate data from two distinct copula models: one governing the binary occurrence process and one governing the positive response component. 
We denote by $C_{\theta_1}$ the copula linking the binary indicator $Z=\mathbbm{1}(Y>0)$ with the covariate $X$, and by $D_{\theta_2}$ the copula linking the positive response $Y\mid Y>0$ with its concomitant covariate. 
Three combinations of copula families are examined:
\begin{itemize}
	\item {GC}: Gaussian copula for the binary component and Clayton copula for the positive component;
	\item {GF}: Gaussian copula for the binary component and  Frank copula for the positive component;
	\item {CF}: Clayton copula for the binary component and  Frank copula for the positive component.
\end{itemize}
For each case, the copula parameters $(\theta_1,\theta_2)$ are chosen to represent moderate positive or negative dependence.

As a benchmark, we also generate data from a logistic/linear quantile regression model following \citet{Ling2022}. 
Specifically, the probability of a positive response and the conditional quantile of the positive component are given by
\begin{equation*}
	\mathrm{logit}\{P(Y>0\mid X)\} = \gamma_0 + \gamma_1 X,
\end{equation*}
and
\begin{equation*}
	Q_{Y\mid Y>0}(\tau\mid X) = \beta_0(\tau) + \beta_1(\tau) X,
\end{equation*}
where two specifications are considered:
\begin{itemize}
	\item {LL1}: logistic and linear quantile regression with $\beta_0(\tau)=\sin(2\pi \tau)$ and $\beta_1(\tau)=\sin(2\pi\tau)+\tau^3$,
	\item {LL2}: logistic and linear quantile regression with $\beta_0(\tau)=\sin(2\pi \tau)$ and  $\beta_1(\tau)=\cos(2\pi\tau)+\tau^3$.
\end{itemize}
Table~\ref{table:DGP} summarizes all data generating processes considered in the simulation study. 
Throughout, the covariate $X$ and the positive response $Y\mid Y>0$ follow Kumaraswamy distributions, denoted by $\mathrm{Kumar}(a,b)$, with shape parameters specified in the table.

\begin{table}[h]
	\setlength{\tabcolsep}{6pt}
	\centering
	\caption{Models for data generating processes.}
	\label{table:DGP}
	\scalebox{0.9}{ 
		\begin{tabular}[t]{lcccccccccccc}
			\toprule[1.5pt]
			&DGP&$C$& $D$&$\{\theta_1, \theta_2\}$&$F_X$&$F_Y$\\

			\midrule
			\multirow{3}{*}{Copula-based}
			&GC&Gaussian&Clayton&$\{0.5, 0.5\}$&$\text{Kumar}(5, 16)$&$\text{Kumar}(2, 5)$\\		
						&GF&Gaussian&Frank&$\{-0.5, -0.5\}$&$\text{Kumar}(5, 16)$&$\text{Kumar}(2, 5)$\\	
									&CF&Clayton&Frank&$\{0.5, -0.5\}$&$\text{Kumar}(5, 16)$&$\text{Kumar}(2, 5)$\\	
			\midrule
			&DGP& $\{\gamma_0, \gamma_1\}$&$\beta_0(\tau)$&$\beta_1(\tau)$&---&---\\
			\midrule
			\multirow{2}{*}{\small Logistic/linear quantile}
			&LL1&$\{0, 0.5\}$&$\sin(2\pi\tau)$&$\sin(2\pi\tau)+\tau^3$&---&---\\		
			&LL2&$\{0, 0.5\}$&$\sin(2\pi \tau)$&$\cos(2\pi\tau)+\tau^3$&---&---\\						
									
			\bottomrule[1.5pt]
		\end{tabular}
	} 
	
	%	{\raggedright  \scriptsize \hspace{0.5ex} The bold values indicate a different decision for Bernstein tests comparing  with other tests under $5\%$ nominal level. \par}
	
	\label{table:240911}
\end{table}

Simulating semicontinuous data with covariate-dependent occurrence and magnitude components poses a nontrivial challenge, as the binary indicator $Z=\mathbbm{1}(Y>0)$ and the positive response $Y\mid Y>0$ must both depend on the covariate $X$ while allowing for distinct dependence structures. 
We therefore adopt a copula-based data generation scheme that separately constructs the binary and positive components through probability integral transforms and conditional copula inversions.
This procedure ensures that the marginal distributions of $X$ and $Y\mid Y>0$ are correctly specified, while the dependence between $(Z,X)$ and $(Y\mid Y>0,X)$ is governed by the copulas $C$ and $D$, respectively. 
Algorithm~\ref{alg:generate_data} summarizes the resulting data generation process. All simulations were conducted in \textsf{R} \citep{Rsoftware}. The function \texttt{BiCopHinv2}, used to compute conditional copula inverses, is implemented in the \textbf{VineCopula} package \citep{VineCoppackage}.

\begin{algorithm}[h]
	\caption{Generate Semicontinuous Data for Given Copulas and Margins}
	\label{alg:generate_data}
	\begin{algorithmic}[1]
		\Input Sample size $n$; Copula parameter $par1$ for copula $C$; Copula parameter $par2$ for copula $D$; Probability of zero response $p_0$
		\Output Sample with variables $X$, $Z$, and $Y$
		\State Generate $n$ observations $\{(U_i, V_i)\}_{i=1}^n$ from copula $C$
		\For{each $i$ in $1, \ldots, n$}
		\State Compute $Z_i$ as an indicator function $Z_i = \II(V_i \geq p_0)$
		\EndFor
		\State Obtain $\{X_i\}_{i=1}^n$ using the inverse CDF of $F_X$ and $U_i$
		\State Extract $X_i$ values where $Z_i = 1$ and denote them as $\{X^+_i\}$
		\State Transform the extracted $\{X^+_i\}$ values to $\{U^+_i\}$ using $F_X$
		\State Generate $t_i$ from a uniform $(0, 1)$ distribution for each extracted $\{X^+_i\}$
		\For{each pair $(U^+_i, t_i)$}
		\State Compute $V^+_i \gets \texttt{BiCopHinv2}\, (t_i, U^+_i, \text{family}, \text{par} = par2)$
		\EndFor
		\State Transform $\{V^+_i\}$ to $\{Y^+_i\}$ using the inverse CDF of $F_Y$
		\State Compute $Y_i \gets \text{ifelse}(Z_i = 1, Y^+_i, 0)$ for each $i$
		\State \Return Data frame with variables $X$, $Z$, and $Y$
	\end{algorithmic}
\end{algorithm}

\subsection{Finite sample performance}

\n This subsection examines the finite-sample performance of the proposed copula-based two-part quantile regression estimator under the data generating processes described in Section~3.1. 
Estimation accuracy is assessed over a range of quantile levels and zero-inflation probabilities, with performance measured using the integrated mean squared error (IMSE) and its decomposition into integrated squared bias (IBIAS$^2$) and integrated variance (IVAR). 
Results are reported for different sample sizes and copula configurations, and are compared with the competing logistic/linear quantile regression estimator.

Tables~\ref{table:n100}--\ref{table:n400} summarize the IMSE, IBIAS$^2$, and IVAR of the estimated conditional quantile functions for increasing sample sizes. 
When the data are generated from copula-based models, the proposed estimator consistently outperforms the logistic/linear quantile regression competitor across all quantile levels and zero-inflation settings. 
The improvement is particularly pronounced at higher quantiles, where the competing estimator exhibits substantial variance inflation, reflecting its inability to accommodate nonlinear dependence structures. 
As the sample size increases, the IMSE of the proposed estimator decreases steadily, with reductions driven primarily by variance shrinkage, while bias remains well controlled.

Notably, the relative efficiency gains of the proposed estimator persist even when the copula families governing the binary and positive components differ, indicating robustness to heterogeneous dependence structures. 
In contrast, the logistic/linear quantile regression estimator performs comparably only in low-dependence settings and moderate quantile levels, but deteriorates rapidly as dependence strengthens or as the quantile level approaches the upper tail. 
These findings highlight the trade-off between model flexibility and stability, and demonstrate that explicitly modeling dependence via copulas yields substantial finite-sample benefits in semicontinuous settings.

\begin{table}[H]
	\setlength{\tabcolsep}{6pt}
	\centering
	\caption{Summary of IMSE, $\rm{IBIAS}^2$, and IVAR of the estimated conditional
		quantile functions on the entire $Q_Y ( \tau|x)$ based on 500 samples, each with 100 observations and Kumaraswamy margin with shape parameters $a=5, b=16$ for covariate,  Kumaraswamy margin with shape parameters $a=2, b=5$ for response. All the vlaues are scaled by ($\times 10^3$)}
	
	\scalebox{0.8}{ 
		\begin{tabular}[t]{lcccccccccccccccccc}
			\toprule[1.5pt]
			\multirow{2}{*}{}\multirow{2}{*}{DGP}&&\multirow{2}{*}{$p_0$}&\multirow{2}{*}{$\tau$}&\multicolumn{3}{c}{Proposed}&
			\multicolumn{3}{c}{ZILQR}&\multirow{2}{*}{\small Ratio of IMSE}\\
			\cmidrule(lr){5-7}
			 \cmidrule(lr){8-10} 
			&&&&IMSE&$\rm{IBIAS}^2$&IVAR&IMSE&$\rm{IBIAS}^2$&IVAR\\

			\midrule
			
			\multirow{9}{*}{GC}
			&& \multirow{3}{*}{$0.1$}&50\%&0.95&0.08&0.87&12.98&5.36&7.62&85.7\%\\
			&&&70\%&1.76&0.38&1.38&26.24&12.99&13.25&50.0\%\\
			&&&90\%&6.72&3.16&3.56&52.55&32.07&20.48&39.6\%\\
            \cline{2-11}
            && \multirow{3}{*}{$0.2$}&50\%&1.12&0.11&1.01&8.98&3.57&5.41&85.7\%\\
			&&&70\%&1.55&0.22&1.33&20.49&9.07&11.42&50.0\%\\
			&&&90\%&5.30&2.26&3.04&48.00&25.86&22.14&39.6\%\\
              \cline{2-11}
			&&\multirow{3}{*}{ 0.4}&50\%&1.35&0.14&1.21 &6.23&2.10&4.13&125.7\%\\
			&&&70\%&1.79&0.18&1.61&12.31 &4.83&7.48&82.1\%\\
			&&&90\%&4.00& 1.28&2.72&29.67&11.68&17.99&48.5\%\\

						\midrule
			
			\multirow{9}{*}{GF}
			&&
			\multirow{3}{*}{$0.1$}&50\%&1.90&0.14&1.76 &10.99&2.35&8.64&66.6\%\\
			&&&70\%&5.31&2.15&3.16&18.07&5.03&13.04&50.0\%\\
			&&&90\%&20.34&12.37&7.97&32.95&13.90&19.05&38.8\%\\
             \cline{2-11}
			&&\multirow{3}{*}{ 0.2}&50\%&1.57&0.13&1.44 &9.24&2.15&7.09&66.6\%\\
			&&&70\%&2.64&0.23&2.41&16.18&4.36&11.82&50.0\%\\
			&&&90\%&10.59&5.36&5.23&33.14&13.56&19.58&38.8\%\\
             \cline{2-11}
 			&&\multirow{3}{*}{ 0.4}&50\%&1.29&0.07&1.22 &7.37&1.71&5.66&66.6\%\\
			&&&70\%&2.22&1.11&2.11&13.72&3.75&9.97&50.0\%\\
			&&&90\%&4.29&0.49& 3.80&33.35&11.76&21.59&38.8\%\\

						\midrule
			
						\multirow{9}{*}{CF}
         &&
		\multirow{3}{*}{$0.1$}&50\%&3.93&0.48 &3.45&14.01& 3.89&10.12 &36.0\%\\
		&&&70\%&7.13&1.24&5.89&24.30&8.46&15.84&23.0\%\\
		&&&90\%& 21.20&10.43&10.77&41.44&20.98&20.46&21.9\%\\
                \cline{2-11}
		&&\multirow{3}{*}{ 0.2}&50\%&2.26&0.29&1.97&10.53&2.97&7.56&0.06\%\\
		&&&70\%&4.26&0.45&3.81&23.43&8.27&15.16&11.3\%\\
		&&&90\%&13.08&5.22&7.86&46.71&25.57&21.14&17.9\%\\
		              \cline{2-11}
			&&\multirow{3}{*}{0.4}&50\%&1.62&0.05&1.57 &10.72&5.08&5.64&125.7\%\\
			&&&70\%&3.60&0.33&3.27&21.72&7.89&13.83&82.1\%\\
			&&&90\%&8.65&1.64&7.01&53.26&27.14&26.12&48.5\%\\

			\bottomrule[1.5pt]
		\end{tabular}
	} 
	
%	{\raggedright  \scriptsize \hspace{0.5ex} The bold values indicate a different decision for Bernstein tests comparing  with other tests under $5\%$ nominal level. \par}
	
	\label{table:n100}
\end{table}

\begin{table}[H]
	\setlength{\tabcolsep}{6pt}
	\centering
	\caption{Summary of IMSE, $\rm{IBIAS}^2$, and IVAR of the estimated conditional
		quantile functions on the entire $Q_Y ( \tau|x)$ based on 500 samples, each with 200 observations and Kumaraswamy margin with shape parameters $a=5, b=16$ for covariate,  Kumaraswamy margin with shape parameters $a=1, b=1$ for response. All the vlaues are scaled by ($\times 10^3$)}
	
	\scalebox{0.8}{ 
		\begin{tabular}[t]{lcccccccccccccccccc}
			\toprule[1.5pt]
			\multirow{2}{*}{}\multirow{2}{*}{DGP}&&\multirow{2}{*}{$p_0$}&\multirow{2}{*}{$\tau$}&\multicolumn{3}{c}{Proposed}&
			\multicolumn{3}{c}{ZILQR}&\multirow{2}{*}{\small Ratio of IMSE}\\
			\cmidrule(lr){5-7}
			 \cmidrule(lr){8-10} 
			&&&&IMSE&$\rm{IBIAS}^2$&IVAR&IMSE&$\rm{IBIAS}^2$&IVAR\\

			\midrule
			
			\multirow{9}{*}{GC}
			&& \multirow{3}{*}{$0.1$}&50\%&0.54&0.08&0.46&12.73&7.39&5.34&85.7\%\\
			&&&70\%&0.96&0.21&0.75&27.78&18.01&9.77&50.0\%\\
			&&&90\%&3.98&1.74&2.24&62.36&48.69&13.67&39.6\%\\
             \cline{2-11}
            && \multirow{3}{*}{$0.2$}&50\%&0.65&0.11&0.54&8.04&4.61&3.43&85.7\%\\
			&&&70\%&0.90&0.19&0.71&19.52&11.81&7.71&50.0\%\\
			&&&90\%&3.15&1.32&1.83&52.03&36.41&15.62&39.6\%\\
             \cline{2-11}
			&&\multirow{3}{*}{ 0.4}&50\%&0.79&0.16&0.63 &4.80&2.42&2.38&125.7\%\\
			&&&70\%&1.10&0.23&0.87&11.11&6.17&4.94&82.1\%\\
			&&&90\%&2.30& 0.78&1.52&30.98&18.31&12.67&48.5\%\\

						\midrule
			
			\multirow{9}{*}{GF}
			&&
			\multirow{3}{*}{$0.1$}&50\%&1.16&0.14&1.02 &7.38&2.33&5.05&66.6\%\\
			&&&70\%&3.15&1.39&1.76&15.36&6.03&9.33&50.0\%\\
			&&&90\%&15.68&10.87&4.81& 29.94& 18.78&11.16&38.8\%\\
            \cline{2-11}
			&&\multirow{3}{*}{ 0.2}&50\%&0.83&0.08&0.75 &6.15&2.57&3.58&50.0\%\\
			&&&70\%&1.64&0.20&1.44&12.91&5.52&7.39&51.8\%\\
			&&&90\%&6.73&3.44& 3.29&32.37&19.64&12.73&60.0\%\\
             \cline{2-11}
			&&\multirow{3}{*}{ 0.4}&50\%&0.71&0.08&0.63 &5.03&2.05&2.98&50.0\%\\
			&&&70\%&1.20&0.11&1.09&9.60&4.22&5.38&51.8\%\\
			&&&90\%&2.76& 0.30& 2.46&27.15&14.92&12.23&60.0\%\\
			
						\midrule
			
						\multirow{9}{*}{CF}
&&
		\multirow{3}{*}{$0.1$}&50\%&1.95&0.26 &1.69&10.67&4.17&6.50 &36.0\%\\
		&&&70\%&3.52&0.43&3.09&22.04&10.34&11.70&23.0\%\\
		&&&90\%&12.68&6.29&6.39&44.83&31.29&13.54&21.9\%\\
               \cline{2-11}
		&&\multirow{3}{*}{ 0.2}&50\%&1.49&0.21&1.28&8.92&3.78&5.14&0.06\%\\
		&&&70\%&2.81&0.18&2.63&21.82&10.82&11.00&11.3\%\\
		&&&90\%&7.33&1.74&5.59&52.78&38.76&14.02&17.9\%\\
		             \cline{2-11}
			&&\multirow{3}{*}{ 0.4}&50\%&0.91&0.02& 0.89&8.91&5.65&3.26&50.0\%\\
			&&&70\%&2.46&0.31&2.15&18.85&9.33&9.52&51.8\%\\
			&&&90\%&5.97&0.76&5.21&59.40&39.82&19.58&60.0\%\\

			\bottomrule[1.5pt]
		\end{tabular}
	} 
	
%	{\raggedright  \scriptsize \hspace{0.5ex} The bold values indicate a different decision for Bernstein tests comparing  with other tests under $5\%$ nominal level. \par}
	
	\label{table:n200}
\end{table}

\begin{table}[H]
	\setlength{\tabcolsep}{6pt}
	\centering
	\caption{Summary of IMSE, $\rm{IBIAS}^2$, and IVAR of the estimated conditional
		quantile functions on the entire $Q_Y ( \tau|x)$ based on 500 samples, each with 400 observations and Kumaraswamy margin with shape parameters $a=5, b=16$ for covariate,  Kumaraswamy margin with shape parameters $a=1, b=1$ for response. All the vlaues are scaled by ($\times 10^3$)}
	
	\scalebox{0.8}{ 
		\begin{tabular}[t]{lcccccccccccccccccc}
			\toprule[1.5pt]
			\multirow{2}{*}{}\multirow{2}{*}{DGP}&&\multirow{2}{*}{$p_0$}&\multirow{2}{*}{$\tau$}&\multicolumn{3}{c}{Proposed}&
			\multicolumn{3}{c}{ZILQR}&\multirow{2}{*}{\small Ratio of IMSE}\\
			\cmidrule(lr){5-7}
			 \cmidrule(lr){8-10} 
			&&&&IMSE&$\rm{IBIAS}^2$&IVAR&IMSE&$\rm{IBIAS}^2$&IVAR\\

			\midrule
			
			\multirow{9}{*}{GC}
			&& \multirow{3}{*}{$0.1$}&50\%&0.29&0.06&0.23&11.65&8.26&3.39&85.7\%\\
			&&&70\%&0.53&0.16&0.37&27.02&20.04&6.98&50.0\%\\
			&&&90\%&2.45&1.15&1.30&64.81&55.30&9.51&39.6\%\\
              \cline{2-11}
            && \multirow{3}{*}{$0.2$}&50\%&0.37&0.08&0.29&7.51&5.24&2.27&85.7\%\\
			&&&70\%&0.51&0.15&0.36&17.87&12.94&4.93&50.0\%\\
			&&&90\%&1.87&0.91&0.96&51.21&40.59&10.62&39.6\%\\
              \cline{2-11}
			&&\multirow{3}{*}{0.4}&50\%&0.45&0.13&0.32 &4.27&3.12&1.15&125.7\%\\
			&&&70\%&0.61&0.17&0.44&10.11&7.17&2.94&82.1\%\\
			&&&90\%&1.41& 0.63&0.78&30.54&22.38&8.16&48.5\%\\

						\midrule
			
			\multirow{9}{*}{GF}
			&&
			\multirow{3}{*}{$0.1$}&50\%&0.73&0.07&0.66 &6.03&2.73&3.30&66.6\%\\
			&&&70\%&1.99& 0.78&1.21&13.95 &7.19& 6.76&50.0\%\\
			&&&90\%&12.68&9.11 &3.57&32.72&25.10 &7.62&38.8\%\\
             \cline{2-11}
			&&\multirow{3}{*}{ 0.2}&50\%&0.45&0.03&0.42 &5.32&3.03&2.29&50.0\%\\
			&&&70\%&1.01&0.11&0.90&11.50&6.59&4.91&51.8\%\\
			&&&90\%&4.55&2.03&2.52&34.77&25.42&9.35&60.0\%\\
              \cline{2-11}
			&&\multirow{3}{*}{0.4}&50\%&0.36&0.03&0.33 &4.16&2.71&1.45&125.7\%\\
			&&&70\%&0.60&0.04&0.56&8.69&5.50&3.19&82.1\%\\
			&&&90\%&1.65& 0.16&1.49&27.68&19.41&8.27&48.5\%\\
			
						\midrule
			
						\multirow{9}{*}{CF}
         &&
		\multirow{3}{*}{$0.1$}&50\%&1.04&0.09 &0.95&8.45&4.43&4.02 &36.0\%\\
		&&&70\%&2.24&0.35&1.89&18.82&10.61&8.21&23.0\%\\
		&&&90\%&7.97&3.46&4.51&47.56&38.83&8.73&21.9\%\\
                \cline{2-11}
		&&\multirow{3}{*}{ 0.2}&50\%&0.72&0.08&0.64&6.99&4.06&2.93&0.06\%\\
		&&&70\%&1.73&0.25&1.48&18.54&11.08&7.46&11.3\%\\
		&&&90\%&4.30&0.68&3.62&53.78&44.14&9.64&17.9\%\\
		              \cline{2-11}
			&&\multirow{3}{*}{0.4}&50\%&0.45&0.01&0.44 &8.35&6.43&1.92&125.7\%\\
			&&&70\%&1.22&0.14&1.08&17.66&11.61&6.05&82.1\%\\
			&&&90\%&3.26&0.23&3.03&62.10&48.89&13.21&48.5\%\\

			\bottomrule[1.5pt]
		\end{tabular}
	} 
	
%	{\raggedright  \scriptsize \hspace{0.5ex} The bold values indicate a different decision for Bernstein tests comparing  with other tests under $5\%$ nominal level. \par}
	
	\label{table:n400}
\end{table}

\subsection{Misspecifications and robustness}

\n This subsection investigates the robustness of the proposed copula-based estimator under model misspecification error.  In particular, we examine the impact of fitting incorrect copula families for the binary and positive components, as well as departures from the copula-based data generating mechanism. 
The objective is to assess the stability of the estimator when the assumed dependence structure does not coincide with the true underlying model, a scenario commonly encountered in practice.

Table~\ref{table:LL} evaluates the robustness of the proposed estimator under two distinct forms of model misspecification. 
When the data are generated from copula-based models but an incorrect copula family is used for estimation, a moderate loss of efficiency is observed, as indicated by increased IMSE values. 
Despite this loss, the proposed estimator remains numerically stable across all quantile levels, with performance degradation driven primarily by variance inflation rather than systematic bias. 
This behavior suggests that the underlying two-part decomposition and piecewise quantile structure are preserved even when the assumed dependence model is misspecified.

In addition, when the data are generated from the logistic/linear quantile regression models LL1 and LL2, the proposed estimator continues to perform competitively and, in several scenarios, attains lower IMSE than the logistic/linear quantile regression estimator, despite the latter being correctly specified. 
This advantage is most pronounced under the LL2 design, where the nonlinear quantile slope induces complex dependence patterns that are not adequately captured by a linear conditional quantile formulation. 
As a consequence, the competing estimator exhibits substantial variance inflation, whereas the copula-based estimator adapts to the induced dependence through its flexible conditional distribution structure, resulting in improved finite-sample stability even under structural misspecification.

Considered collectively, these results indicate that the proposed copula-based framework achieves a favorable balance between modeling flexibility and robustness, making it particularly well suited for the analysis of semicontinuous data arising in applications with complex and potentially misspecified dependence structures.

\begin{table}[H]
	\setlength{\tabcolsep}{6pt}
	\centering
	\caption{Summary of IMSE, $\rm{IBIAS}^2$, and IVAR of the estimated conditional
		quantile functions on the entire $Q_Y ( \tau|x)$ based on 500 samples, each with 400 observations and Kumaraswamy margin with shape parameters $a=5, b=16$ for covariate,  Kumaraswamy margin with shape parameters $a=2, b=5$ for response. All the vlaues are scaled by ($\times 10^3$)}
	
	\scalebox{0.8}{ 
		\begin{tabular}[t]{lcccccccccccccccccc}
			\toprule[1.5pt]
			\multirow{2}{*}{}\multirow{2}{*}{DGP}& \multirow{2}{*}{Fits}&\multirow{2}{*}{$p_0$}&\multirow{2}{*}{$\tau$}&\multicolumn{3}{c}{Proposed}&
			\multicolumn{3}{c}{ZILQR}&\multirow{2}{*}{\small Ratio of IMSE}\\
			\cmidrule(lr){5-7}
			 \cmidrule(lr){8-10} 
			&&&&IMSE&$\rm{IBIAS}^2$&IVAR&IMSE&$\rm{IBIAS}^2$&IVAR\\
			\midrule

			\multirow{18}{*}{GC}
			&\multirow{9}{*}{FC}&
			\multirow{3}{*}{$0.1$}&50\%&0.51&0.28&0.23 &11.65&8.26&3.39&85.7\%\\
			&&&70\%&0.60&0.19&0.41&27.02&20.04&6.98&50.0\%\\
			&&&90\%&2.17&0.65&1.52&64.81&55.30&9.51&39.6\%\\
             \cline{3-11}
			&&\multirow{3}{*}{ 0.2}&50\%&0.58&0.18&0.40 &7.51&5.24&2.27&85.7\%\\
			&&&70\%&0.80&0.42&0.38&17.87&12.94&4.93&50.0\%\\
			&&&90\%&1.78&0.52&1.26&51.21&40.59&10.62&39.6\%\\
             \cline{3-11}
 			&&\multirow{3}{*}{ 0.4}&50\%&0.53&0.15&0.38 &4.27&3.12&1.15&125.7\%\\
			&&&70\%&0.89&0.32&0.57&10.11&7.17&2.94&82.1\%\\
			&&&90\%&1.74&0.79&0.95&30.54&22.38&8.16&48.5\%\\
            \cline{3-11}
&\multirow{9}{*}{GF}&
			\multirow{3}{*}{$0.1$}&50\%&1.81&1.16&0.65 &11.65&8.26&3.39&85.7\%\\
			&&&70\%&3.54&2.27&1.27&27.02&20.04&6.98&50.0\%\\
			&&&90\%&5.77&2.05&3.72&64.81&55.30&9.51&39.6\%\\
             \cline{3-11}
			&&\multirow{3}{*}{ 0.2}&50\%&1.27&0.84&0.43 &7.51&5.24&2.27&85.7\%\\
			&&&70\%&2.30&1.37&0.93&17.87&12.94&4.93&50.0\%\\
			&&&90\%&4.96&2.39&2.57&51.21&40.59&10.62&39.6\%\\
             \cline{3-11}
 			&&\multirow{3}{*}{ 0.4}&50\%&1.34&1.04&0.30 &4.27&3.12&1.15&125.7\%\\
			&&&70\%&1.57&0.99&0.58&10.11&7.17&2.94&82.1\%\\
			&&&90\%&3.05&1.44&1.61&30.54&22.38&8.16&48.5\%\\

									\midrule
			
						\multirow{3}{*}{LL1}
						
%			&$0.5$-$(\log(\tau), \tau^3)$&---&0.322&0.183&0.139&0.028& 0.007&0.021&1150.0\%\\
&\multirow{3}{*}{LL1}&\multirow{3}{*}{---}&50\%&913.05&853.95&59.10&2006.31&1986.29&20.02& 132.0\%\\
			&&&70\%&411.95&303.78&108.17&1307.34&1286.45&20.89&401.0\%\\
			&&&90\%&1886.51&1796.42&90.09&1909.52&1899.29&10.23&112.7\%\\
            \midrule
									\multirow{3}{*}{LL2}
&\multirow{3}{*}{LL2}&\multirow{3}{*}{---}&50\%&755.01&670.04&84.97&1938.03&1901.74&36.29&43.4\%\\
			&&&70\%&593.67&438.34&155.33&1243.82&1205.24&38.58&32.8\%\\
			&&&90\%&3362.00&3235.96&126.24&4106.41&4077.58&28.83&77.5\%\\

			\bottomrule[1.5pt]
		\end{tabular}
	} 
	
%	{\raggedright  \scriptsize \hspace{0.5ex} The bold values indicate a different decision for Bernstein tests comparing  with other tests under $5\%$ nominal level. \par}
	
	\label{table:LL}
\end{table}

%---------------------------------------------------------------------	
\section{Healthcare data analysis}\label{sec:application}
%---------------------------------------------------------------------	

\n This section illustrates the practical performance of the proposed copula-based two-part quantile regression framework through an application to a real-world semicontinuous outcome. 
The analysis aims to demonstrate how explicitly modeling both the occurrence and magnitude of positive responses, together with flexible dependence structures, yields a nuanced characterization of conditional quantile effects that cannot be captured by standard two-part or linear quantile regression approaches. 
The empirical results highlight the interpretability and stability of the proposed method in a realistic setting involving substantial zero inflation and heterogeneous covariate effects.

\subsection{Motivating dataset}

\n The healthcare dataset analyzed in this study comprises three variables: uncompensated care burden, charity burden, and the weighted Social Deprivation Index (SDI) score, encompassing a total of 4,700 observations.

\subsection{Model results}

\noindent A primary objective of this analysis is to investigate how uncompensated care burden and charity care burden vary with changes in the weighted Social Deprivation Index (SDI) score. 
To this end, we consider three modeling approaches:
\begin{enumerate}[(i)]
    \item direct-fit linear quantile regression,
    \item logistic/linear quantile regression, and
    \item the proposed copula-based two-part quantile regression.
\end{enumerate}
For each outcome, the estimated probabilities of occurrence and the conditional quantiles at the $50\%$, $70\%$, and $90\%$ levels are displayed in \hyperref[fig:uncomp]{Fig.~\ref{fig:uncomp}} and \hyperref[fig:charity]{Fig.~\ref{fig:charity}}. 
Pointwise $95\%$ confidence bands, constructed using 300 bootstrap replications, are also shown to quantify estimation uncertainty.

When comparing the proposed copula-based quantile regression model with the two linear quantile regression approaches, several important differences emerge. 
The linear structure imposed by the competing models limits their ability to capture the nonlinear relationships evident in the data, particularly in the presence of substantial zero inflation. 
In this application, the direct-fit linear quantile regression model systematically yields smaller estimated slopes and narrower conditional distributions than its two-part counterpart, reflecting the inadequacy of a single linear specification for modeling semicontinuous healthcare outcomes. 
By contrast, the proposed copula-based framework accommodates nonlinear dependence and provides more flexible and stable quantile estimates across the distribution.

As shown in \hyperref[fig:uncomp]{Fig.~\ref{fig:uncomp}} and \hyperref[fig:charity]{Fig.~\ref{fig:charity}}, the association between the Social Deprivation Index (SDI) and both uncompensated and charity care burdens exhibits a concave pattern across quantile levels. 
At low to moderate SDI values, higher social deprivation is associated with increased probabilities and magnitudes of care burden, consistent with the interpretation that socioeconomic disadvantage is linked to greater unmet healthcare needs and financial vulnerability.

At higher SDI levels, however, the estimated probability of experiencing a care burden begins to decline. 
One plausible explanation is that in areas of extreme deprivation, individuals may forgo healthcare services altogether due to severe access barriers, thereby reducing observed utilization and the measured burden of uncompensated or charity care. 
An alternative, nonexclusive explanation is that residents in highly deprived areas may be more likely to qualify for public assistance programs or receive care through safety-net providers, which can partially offset uncompensated care exposure. 
Although the available data do not allow these mechanisms to be directly distinguished, the observed downturn suggests that the relationship between social deprivation and care burden is inherently nonlinear.

Finally, the decreasing trend at higher SDI levels is less pronounced for the $90$th quantile than for the $50$th and $70$th quantiles. 
This pattern is consistent with the interpretation that extreme levels of care burden correspond to a relatively small subset of individuals, for whom reduced utilization or access constraints play a more limited role in determining the conditional magnitude of burden. 
As a result, the upper tail of the distribution remains comparatively stable even as overall care utilization declines in highly deprived settings.

%%%%%%%%%%%%%%%%%%%%%%%%%%%%%%%%%%%%%%%%%
\begin{figure}[H]
	\centering
	\begin{subfigure}{.49\textwidth}
		\makebox[\textwidth][c]{\includegraphics[width=3.4in, height=2.7in]{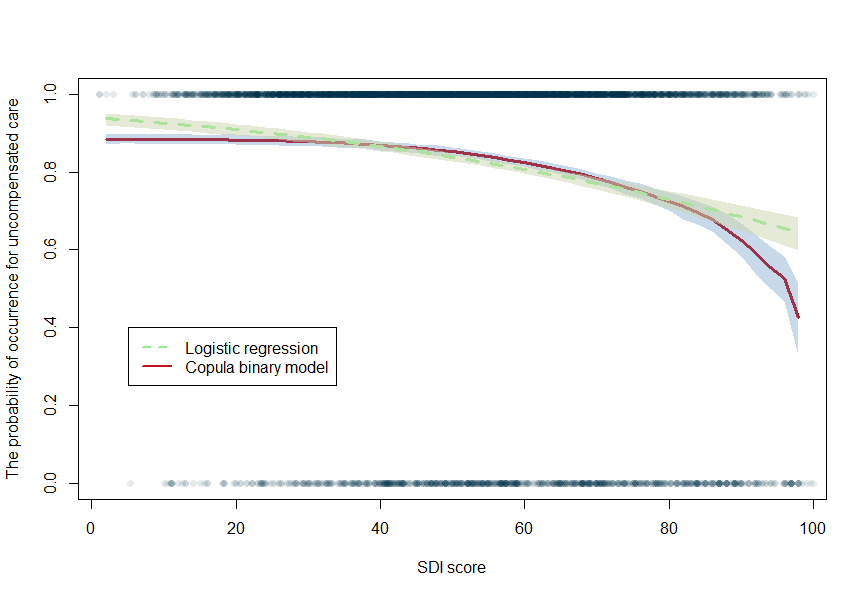}}%	
		\vspace{0.5cm}
	\end{subfigure}
	\begin{subfigure}{.5\textwidth}
		\makebox[\textwidth][c]{\includegraphics[width=3.4in, height=2.7in]{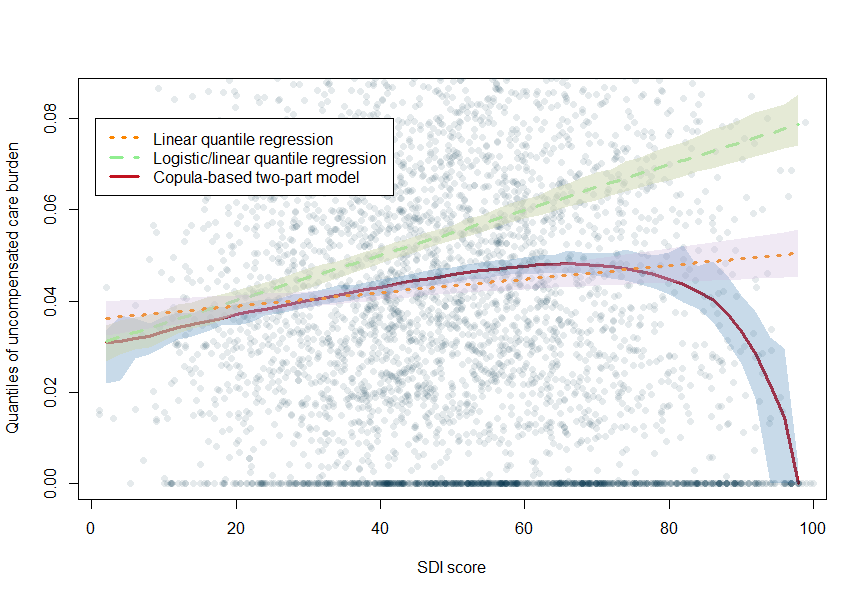}}%	
		\vspace{0.5cm}
	\end{subfigure}
	\begin{subfigure}{.49\textwidth}
	\makebox[\textwidth][c]{\includegraphics[width=3.4in, height=2.7in]{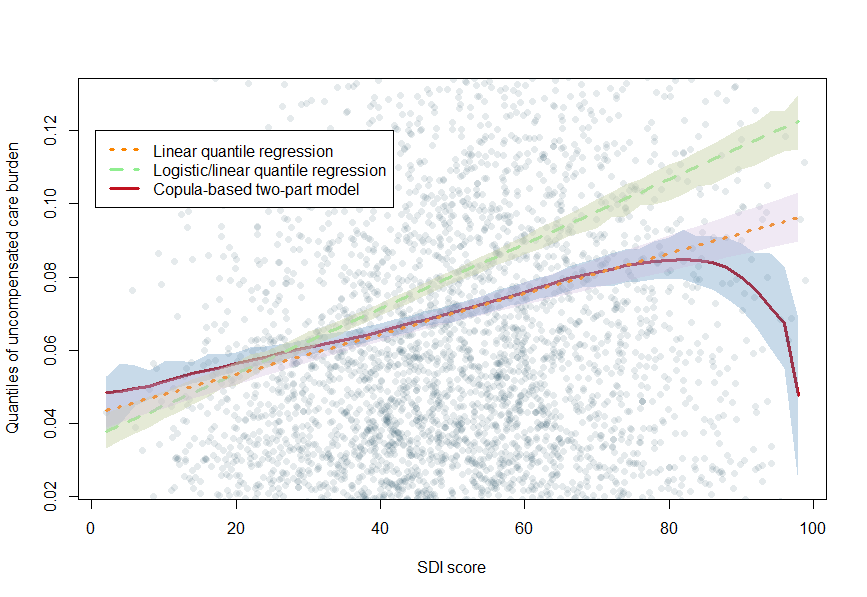}}%	
	\vspace{0.5cm}
\end{subfigure}
\begin{subfigure}{.5\textwidth}
	\makebox[\textwidth][c]{\includegraphics[width=3.4in, height=2.7in]{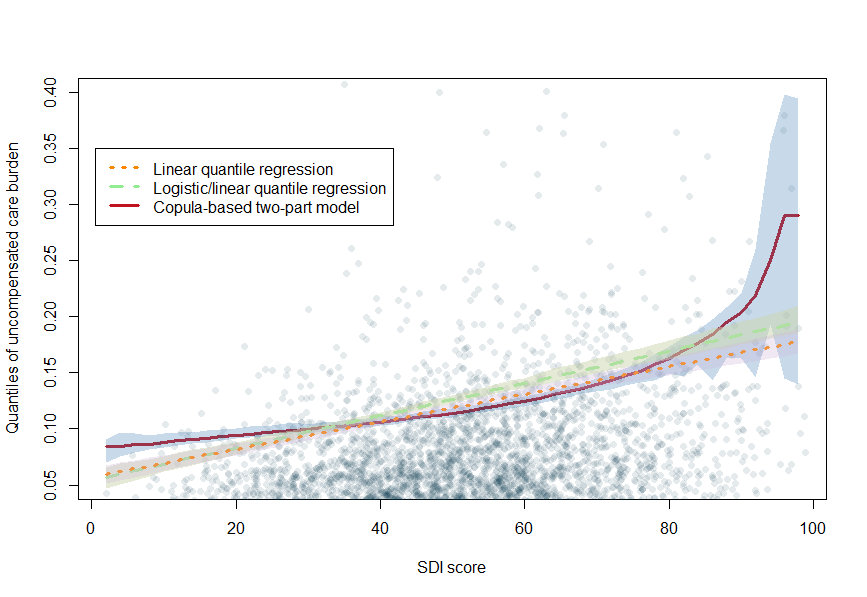}}%	
	\vspace{0.5cm}
\end{subfigure}

	\caption{Uncompensated care burden versus SDI score, all pointwise confidence bands are derived from $300$ bootstrap samples. \textbf{TopLeft:}  The estimated probability of occurrence for uncompensated care with logistic regression (green dash line) and copula binary model (red solid line); \textbf{TopRight:} The quantile estimates of level $50\%$ for linear quantile regression (orange dotted line), logistic/linear quantile regression (green dash line) and copula-based two-part model (red solid line); \textbf{BottomLeft:} The quantile estimates of level $70\%$ for the three models; \textbf{BottomRight:} The quantile estimates of level $90\%$ for the three models.}
	\label{fig:uncomp}
\end{figure}

%%%%%%%%%%%%%%%%%%%%%%%%%%%%%%%%%%%%%%%%%%%%%

%%%%%%%%%%%%%%%%%%%%%%%%%%%%%%%%%%%%%%%%%
\begin{figure}[H]
	\centering
	\begin{subfigure}{.49\textwidth}
		\makebox[\textwidth][c]{\includegraphics[width=3.4in, height=2.7in]{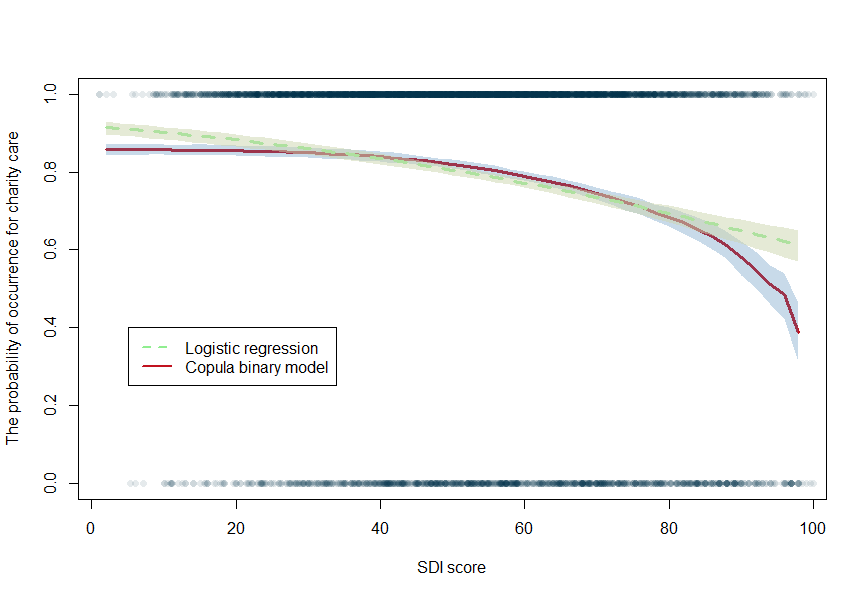}}%	
		\vspace{0.5cm}
	\end{subfigure}
	\begin{subfigure}{.5\textwidth}
		\makebox[\textwidth][c]{\includegraphics[width=3.4in, height=2.7in]{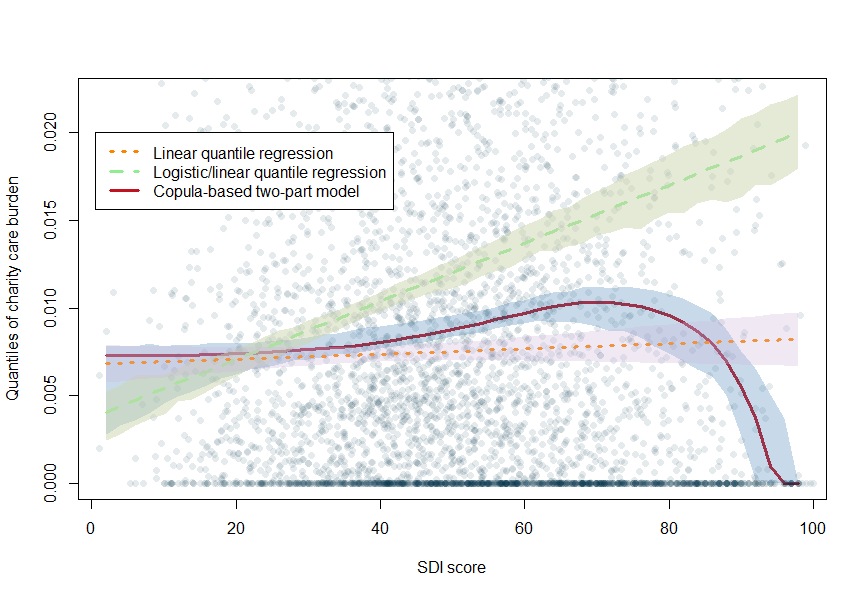}}%	
		\vspace{0.5cm}
	\end{subfigure}
	\begin{subfigure}{.49\textwidth}
	\makebox[\textwidth][c]{\includegraphics[width=3.4in, height=2.7in]{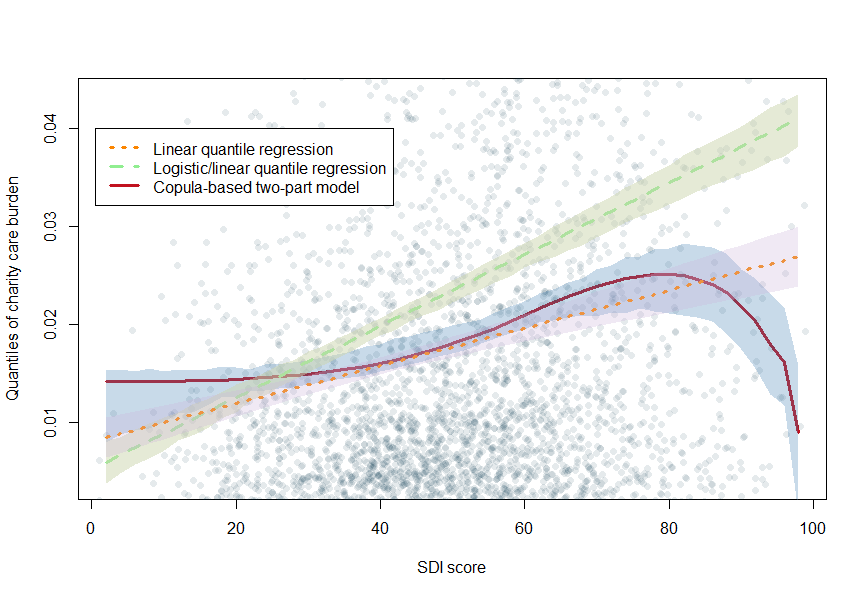}}%	
	\vspace{0.5cm}
\end{subfigure}
\begin{subfigure}{.5\textwidth}
	\makebox[\textwidth][c]{\includegraphics[width=3.4in, height=2.7in]{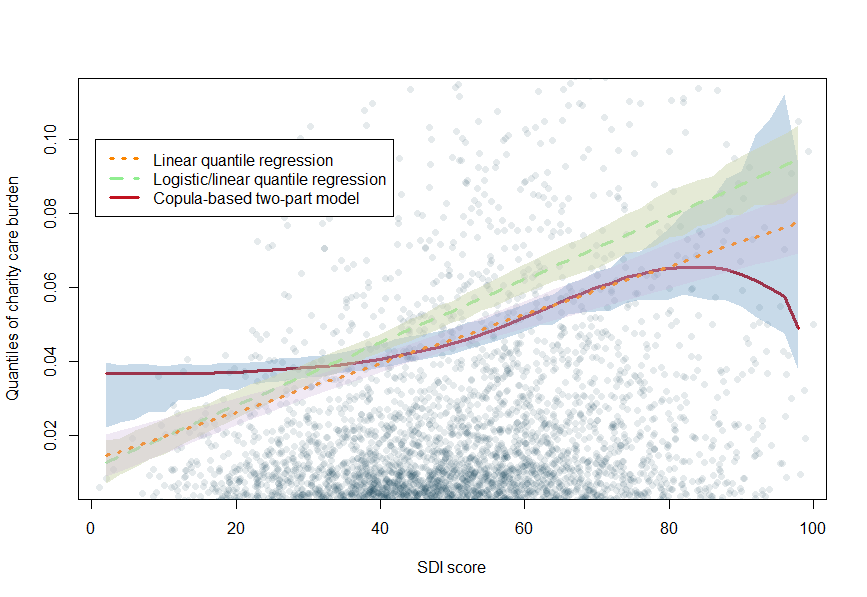}}%	
	\vspace{0.5cm}
\end{subfigure}

	\caption{Charity care burden versus SDI score, all pointwise confidence bands are derived from $300$ bootstrap samples. \textbf{TopLeft:}  The estimated probability of occurrence for charity care with logistic regression (green dash line) and copula binary model (red solid line); \textbf{TopRight:} The quantile estimates of level $50\%$ for linear quantile regression (orange dotted line), logistic/linear quantile regression (green dash line) and copula-based two-part model (red solid line); \textbf{BottomLeft:} The quantile estimates of level $70\%$ for the three models; \textbf{BottomRight:} The quantile estimates of level $90\%$ for the three models.}
	\label{fig:charity}
\end{figure}

%%%%%%%%%%%%%%%%%%%%%%%%%%%%%%%%%%%%%%%%%%%%%

%---------------------------------------------------------------------	
\section{Concluding remarks}\label{sec:conclusion}
%---------------------------------------------------------------------	

\noindent This paper proposes a copula-based semiparametric quantile regression framework for analyzing semicontinuous outcomes. 
The method explicitly accounts for the mixture structure induced by excess zeros and accommodates flexible dependence between the response and covariates. 
Large-sample properties of the proposed estimator are established, and simulation studies demonstrate improved finite-sample performance relative to logistic/linear quantile regression, particularly in settings with nonlinear dependence and substantial zero inflation. 
In an application to healthcare data, the proposed approach yields a more nuanced characterization of the relationship between social deprivation and care burdens, revealing heterogeneous effects across quantiles that are not captured by competing models.

Several directions for future research merit consideration. 
First, the proposed framework can be extended to accommodate censored semicontinuous outcomes, which commonly arise in health services and cost analyses. 
Second, incorporating clustered or longitudinal data structures, such as repeated observations within patients, hospitals, or geographic regions, would allow the model to account for within-unit dependence through hierarchical or copula-based random effects. 
Third, further investigation of extremal quantiles may provide additional insights into tail behavior and risk concentration, particularly for identifying populations at greatest financial vulnerability. 
Finally, extensions to higher-dimensional covariate settings and data-driven copula selection procedures would enhance the scalability and practical applicability of the proposed methodology.

% future work: composite quantile regression, censored quantile regression (only response is censored),

% multiple excess, missing data, random effects model and future work added form Dr. Hussein

\section*{Appendix: Proof of~\hyperref[thm:2024-06-09, 8:46PM]{Theorem~\ref{thm:2024-06-09, 8:46PM}}} \label{app}% Use section* to prevent numbering
\addcontentsline{toc}{section}{Appendix: Proofs} % Add to the table of contents

% \renewcommand{\thesubsection}{A.\arabic{subsection}} % Customize subsection numbering

% \subsection{Proof of~\hyperref[thm:2024-06-09, 8:46PM]{Theorem~\ref{thm:2024-06-09, 8:46PM}}}

\noindent This proof follows the methodology outlined in~\citet[Proof of Theorem 1 (i)]{Ling2022}. We prove the consistency according to the three sub-intervals.
\begin{itemize}
	\item For $\tau < C_{V|\bm{U}}(p_0|F_{\bm{X}}(\bm{x}); \bm{\theta}_1)$. Note that, by~\citet[Proof of Theorem 1]{Mesfioui2023}, under \hyperref[ass:2024-06-10, 10:11AM]{Assumption~\ref{ass:2024-06-10, 10:11AM}} (iii)--(v), one has the weak consistency of estimator ${C}_{V|\bm{U}}\left(\widehat{p}_0|\widehat{F}_{\bm{X}}(\bm{x}); \widehat{\bm{\theta}}_1\right)$, which implies that $P\left(A_{2n} \cup A_{3n}\right)\to 0$. For any $\varepsilon >0$, one has
	\begin{align*}
		P\left(\left|\widehat{Q}(\tau|\bm{x})-0\right|>\varepsilon \right) & = P\left(\left|\widehat{Q}(\tau|\bm{x})\right|>\varepsilon \right) \\
		 & \leq P\left( \left\{\left|\widehat{Q}(\tau|\bm{x})\right|> \varepsilon\right\} \cap A_{1n} \right)  + P\left(A_{2n} \cup A_{3n}\right)\\
		 &=P\left(A_{2n} \cup A_{3n}\right)\to 0.
	\end{align*}
	
	\item For $\tau = C_{V|\bm{U}}(p_0|F_{\bm{X}}(\bm{x}); \bm{\theta}_1)$. Note that, under \hyperref[ass:2024-06-10, 10:11AM]{Assumption~\ref{ass:2024-06-10, 10:11AM}} (iii)--(v), by~\citet[][Theorem 1]{Mesfioui2023}, $\sqrt{n}\left\{C_{V|\bm{U}}(p_0|F_{\bm{X}}(\bm{x}); \bm{\theta}_1)- {C}_{V|\bm{U}}\left(\widehat{p}_0|\widehat{F}_{\bm{X}}(\bm{x}); \widehat{\bm{\theta}}_1\right)\right\}$ is asymptotically normal. Hence, one has 
	\begin{align*}
		P(A_{1n}) &= P\left(\sqrt{n}\left\{C_{V|\bm{U}}(p_0|F_{\bm{X}}(\bm{x}); \bm{\theta}_1)- {C}_{V|\bm{U}}\left(\widehat{p}_0|\widehat{F}_{\bm{X}}(\bm{x}); \widehat{\bm{\theta}}_1\right)\right\} < 0\right) \to \frac{1}{2}, \\
		P(A_{2n}) &=P\left( 0 \leq \sqrt{n}\left\{C_{V|\bm{U}}(p_0|F_{\bm{X}}(\bm{x}); \bm{\theta}_1)- {C}_{V|\bm{U}}\left(\widehat{p}_0|\widehat{F}_{\bm{X}}(\bm{x}); \widehat{\bm{\theta}}_1\right)\right\}  \leq n^{1/2-\delta}\right) \to \frac{1}{2},\\
		P(A_{3n})&\to 0.
	\end{align*}
Therefore, 	
		\begin{align*}
		&P\left(\left|\widehat{Q}(\tau|\bm{x})-0\right|>\varepsilon \right)\\ 
        &\quad = P\left( \left\{\left|\widehat{Q}(\tau|\bm{x})\right|> \varepsilon\right\} \cap A_{1n} \right) +P\left( \left\{\left|\widehat{Q}(\tau|\bm{x})\right|> \varepsilon\right\} \cap A_{2n} \right)\\
		&\qquad  + P\left( \left\{\left|\widehat{Q}(\tau|\bm{x})\right|> \varepsilon\right\} \cap A_{3n} \right)  \\
		& \quad\leq 0 + P\left( \left\{\left|\widehat{Q}(\tau|\bm{x})\right|> \varepsilon\right\} \cap A_{2n} \right)  + P\left(A_{3n}\right)\\
		& \quad = P\left( \left\{\left|\widehat{F}_{Y|Y>0}^{-1}\left({D}_{V^+|\bm{U}^+}^{-1}\left(\widehat{\tau}_s\left({C}_{V|\bm{U}}\left(\widehat{p}_0|\widehat{F}_{\bm{X}}(\bm{x}); \widehat{\bm{\theta}}_1\right)+n^{-\delta}, \bm{x}\right)\big|\widehat{F}_{\bm{X}|Y>0}(\bm{x}); \widehat{\bm{\theta}}_2\right)\right)\right.\right.\right.\\ 
		&\qquad \left.\left.\left.\cdot \frac{C_{V|\bm{U}}(p_0|F_{\bm{X}}(\bm{x});\bm{\theta}_1)-{C}_{V|\bm{U}}\left(\widehat{p}_0|\widehat{F}_{\bm{X}}(\bm{x}); \widehat{\bm{\theta}}_1\right)}{n^{-\delta}}\right|> \varepsilon\right\} \right)  + P\left(A_{3n}\right)\\
		&\quad\to 0,
	\end{align*}
where the limit due to the fact that, under \hyperref[ass:2024-06-10, 10:11AM]{Assumption~\ref{ass:2024-06-10, 10:11AM}} (i), one has, for $0 < \delta < 1/2$, 
	
\begin{align*}
&\widehat{F}_{Y|Y>0}^{-1}\left({D}_{V^+|\bm{U}^+}^{-1}\left(\widehat{\tau}_s\left({C}_{V|\bm{U}}\left(\widehat{p}_0|\widehat{F}_{\bm{X}}(\bm{x}); \widehat{\bm{\theta}}_1\right)+n^{-\delta}, \bm{x}\right)\big|\widehat{F}_{\bm{X}|Y>0}(\bm{x}); \widehat{\bm{\theta}}_2\right)\right)\\
&\quad  \cdot \frac{C_{V|\bm{U}}(p_0|F_{\bm{X}}(\bm{x}); \bm{\theta}_1)-\widehat{C}_{V|\bm{U}}\left(\widehat{p}_0|\widehat{F}_{\bm{X}}(\bm{x});\widehat{\bm{\theta}}_1\right)}{n^{-\delta}} \\
 & \quad = \widehat{F}_{Y|Y>0}^{-1}\left({D}_{V^+|\bm{U}^+}^{-1}\left(\frac{n^{-\delta}}{1-{C}_{V|\bm{U}}\left(\widehat{p}_0|\widehat{F}_{\bm{X}}(\bm{x}); \widehat{\bm{\theta}}_1\right)}\Big| \widehat{F}_{\bm{X}|Y>0}(\bm{x}); \widehat{\bm{\theta}}_2\right)\right)n^{\delta-1/2}\\
 &\qquad  \cdot \left\{n^{1/2} \left[{C_{V|\bm{U}}(p_0|F_{\bm{X}}(\bm{x}); \bm{\theta}_1)-{C}_{V|\bm{U}}\left(\widehat{p}_0|\widehat{F}_{\bm{X}}(\bm{x}); \widehat{\bm{\theta}}_1\right)}\right]\right\}\\
 &\quad = o_p(1)\cdot O_p(1)=o_p(1).
\end{align*}

	\item For $\tau > C_{V|\bm{U}}(p_0|F_{\bm{X}}(\bm{x}); \bm{\theta}_1)$. One has $P(A_{3n})\to 1$ and $P(A_{1n}\cup A_{2n}) \to 0$. Then
	\begin{align}\label{eq:2024-06-10, 1:13PM}
	&P\left(\left|\widehat{Q}(\tau|\bm{x})-Q(\tau |\bm{x})\right|>\varepsilon \right)  \notag\\
 &\quad \leq P(A_{1n}\cup A_{2n}) + P\left(\left|\widehat{F}_{Y|Y>0}^{-1}\left({D}_{V^+|\bm{U}^+}^{-1}\left(\widehat{\tau}_s\left(\tau, \bm{x}\right)\big|\widehat{F}_{\bm{X}|Y>0}(\bm{x}); \widehat{\bm{\theta}}_2\right)\right)-Q(\tau |\bm{x})\right|>\varepsilon \right).
	\end{align}
Under \hyperref[ass:2024-06-10, 10:11AM]{Assumption~\ref{ass:2024-06-10, 10:11AM}} (iv)-(vi), by~\citet[Theorem 2]{Remillard2017}, \\ $\widehat{F}_{Y|Y>0}^{-1}\left({D}_{V^+|\bm{U}^+}^{-1}\left(\widehat{\tau}_s\left(\tau, \bm{x}\right)\big|\widehat{F}_{\bm{X}|Y>0}(\bm{x}); \widehat{\bm{\theta}}_2\right)\right)$ is consistent for $Q(\tau |\bm{x})$.  Therefore, Equation~\eqref{eq:2024-06-10, 1:13PM} converges to zero.
\end{itemize}
This completes the proof. \qed

\bibliographystyle{chicago}
\bibliography{Mybib}	
\end{document}